\documentclass[3p]{elsarticle}
\usepackage{amsmath,amssymb}
\usepackage{tikz}
\usetikzlibrary{positioning,calc}
\usepackage{algorithmicx}
\usepackage[noend]{algpseudocode}
\usepackage{hyperref}

\newtheorem{proposition}{Proposition}
\newtheorem{theorem}{Theorem}
\newproof{proof}{Proof}

\begin{document}
\title{Complexity of Atoms, Combinatorially}
\author[szabivan]{Szabolcs~Iv\'an\tnoteref{thanksfn}}
\address[szabivan]{University of Szeged, Hungary}
\tnotetext[thanksfn]{This research was supported by the European Union and the State of Hungary, co-financed by the European Social Fund in the framework of T\'AMOP 4.2.4.A/2-11-1-2012-0001 ‘National Excellence Program’ and by the OTKA grant no. 108448.}
\begin{abstract}
Atoms of a (regular) language $L$ were introduced by Brzozowski and Tamm in 2011 as intersections of complemented and uncomplemented quotients of $L$.
They derived tight upper bounds on the complexity of atoms in 2013.
In 2014, Brzozowski and Davies characterized the regular languages meeting these bounds.
To achieve these results, they used the so-called ``\'atomaton'' of a language, introduced by Brzozowski and Tamm in 2011.

In this note we give an alternative proof of their characterization, via a purely combinatorial approach.
\end{abstract}
\maketitle
\section{Introduction}
The state complexity of a regular language $L$ is the number of states of its minimal automaton.
An atom of a language is a non-empty intersection of its quotients, some of which may be complemented.
Brzozowski and Tamm introduced atoms in~\cite{theoryatomata} and found tight upper bounds for their
state complexity in~\cite{DBLP:journals/ijfcs/BrzozowskiT13}, carefully analyzing a particular nondeterministic
finite automaton, the so-called ``\'atomaton'' of a regular language also introduced in~\cite{theoryatomata}.

A language is defined to be \emph{maximally atomic} in~\cite{brzozoMaximally} if it has the maximal number of atoms possible and
each of the individual atoms has the maximal possible state complexity.
In~\cite{DBLP:journals/corr/abs-1302-3906}, Brzozowski and Davies showed that maximal syntactic complexity implies
maximal atomicity and in~\cite{brzozoMaximally} they gave necessary and sufficient conditions for a
language to be maximally atomic.

In this paper we introduce another tool which we call the ``disjoint power set automaton'' of a regular language,
and give a self-contained, purely combinatorial and automata-theoretic proof of their characterization.

\section{Notation}
A \emph{semigroup} $(S,\cdot)$ is a set $S$ equipped with a binary associative operation $\cdot$.
We usually omit the sign $\cdot$ and write $st$ for $s\cdot t$.
A \emph{monoid} is a semigroup $(S,\cdot)$ having a neutral element $1$ satisfying $s1=1s=s$ for each $s\in S$.
Given a finite nonempty set $Q$, two particular semigroups are $\mathcal{T}_Q$ consisting of all
the transformations of $Q$ (i.e. functions $Q\to Q$ with function composition as product) and
its subsemigroup $\mathcal{P}_Q$ consisting of the \emph{permutations} of $Q$.
In order to ease notation in the automata theoretic part, we write function application in diagrammatic
order, i.e. if $p\in Q$ and $f\in\mathcal{T}_Q$, then $pf$ stands for the value to which $f$ maps $p$,
and for $f,g\in\mathcal{T}_Q$ their product is $fg$ defined as $p(fg)=(pf)g$ for each $p$.
Also, when $f\in\mathcal{T}_Q$ and $S\subseteq Q$, then $Sf$ stands for the set $\{sf:s\in S\}$.
The \emph{rank} of a transformation $f\in\mathcal{T}_Q$ is the cardinality of its image $Qf$;
transformations of rank $n$ are called \emph{permutations}, while all other transformations are
called \emph{singular} transformations. 
When $n\geq 0$ is an integer, then $\mathcal{T}_n$ stands for the transformation semigroup $\mathcal{T}_{\{1,\ldots,n\}}$.

An \emph{alphabet} is a finite nonempty set $\Sigma$ of symbols.
A \emph{$\Sigma$-word} is a finite sequence $w=a_1a_2\ldots a_n$ with each $a_i$ being in $\Sigma$.
For $n=0$ we get the \emph{empty word}, denoted $\varepsilon$.
The set $\Sigma^*$ of all words forms a monoid with the operation being (con)catenation, or simply
product of words given by $a_1\ldots a_n\cdot b_1\ldots b_k=a_1\ldots a_nb_1\ldots b_k$. In this
monoid, $\varepsilon$ is the neutral element. The semigroup $\Sigma^+=\Sigma^*-\{\varepsilon\}$
is the semigroup of nonempty words.

A \emph{language} (over $\Sigma$) is an arbitrary subset of $\Sigma^*$.
A \emph{finite automaton} is a system $M=(Q,\Sigma,\delta,q_0,F)$ with $Q$ being the finite nonempty set of states,
$\Sigma$ being the input alphabet, $\delta:Q\times\Sigma \to Q$ is the transition function, $q_0\in Q$ is the
start state and $F\subseteq Q$ is the set of final states.
Given $M$, the monoid $\Sigma^*$ acts on $Q$ from the right as $q\cdot_M \varepsilon=q$
and $q\cdot_M ua=\delta(q\cdot_M u,a)$ for each $q\in Q$, $u\in\Sigma^*$ and $a\in\Sigma$.
When $M$ is clear from the context, we omit the subscript and, in most cases, also the period and write only $qw$
for $q\cdot_M w$. Then, each word $w$ \emph{induces} a function $Q\to Q$, denoted by $w_M$, defined as $q\mapsto qw$.
The \emph{transformation semigroup} of $M$ is $\mathcal{T}(M)=\{w_M:w\in\Sigma^+\}$  -- it is clear that $u_Mv_M=(uv)_M$ so
$\mathcal{T}(M)$ is indeed a semigroup. Most of the time, when $M$ is clear, we omit the subscript also here and identify
$w$ with $w_M$.
Another semigroup associated to $M$ is that of its permutation group
$\mathcal{P}(M)=\{w_M:w\in\Sigma^*,Qw=Q\}$.
The language recognized by $M$ is the language $L(M)=\{w:q_0w\in F\}$. A language is called \emph{regular} if it can be
recognized by a finite automaton. It is well-known that for each regular language there exists a \emph{minimal} automaton,
unique up to isomorphism having the least number of states among all the automata recognizing $L$.

A state $q$ of an automaton is \emph{reachable} from a state $p$ if $pw=q$ for some word $w$.
States that are reachable from $q_0$ are simply called \emph{reachable} states.
A \emph{sink} is a non-final state $p\notin F$ such that $pa=p$ for each $a\in \Sigma$ (thus, $pw=p$ for each word $w$ as well).
Two states $p,q$ are called \emph{distinguishable} if there exists a word $w$ such that exactly one of the states $pw$ and $qw$ belongs to $F$.
It is known that $M$ is minimal iff each pair $p\neq q$ of its states is distinguishable and all its states are reachable.
When $M=(Q,\Sigma,\delta,q_0,F)$ is an automaton and $q\in Q$, then $M_q$ stands for the automaton $(Q,\Sigma,\delta,q,F)$
and $L_q$ for the language recognized by $M_q$.
A state $q$ is \emph{empty} if so is $L_q$. In a minimal automaton, there is at most one empty state which is then a sink.
For a subset $S\subseteq Q$ of states, let $L_S$ stand for $\cup_{q\in S}L_q$.

Given a (regular) language $L\subseteq\Sigma^*$, a well-known associated congruence on words is its \emph{syntactic right congruence}
$\sim_L$ defined as
\[x\sim_Ly\ \Leftrightarrow\ (\forall z:\ xz\in L\Leftrightarrow yz\in L).\]
It is known that the minimal automaton of a regular language $L$ is isomorphic to $(\Sigma^*/{\sim_L},\Sigma,\delta_L,\varepsilon/{\sim_L},L/{\sim_L})$
where $\delta_L(x/{\sim_L},a)=xa/{\sim_L}$. 

Similarly\footnote{Though the notion of syntactic left congruence seems to be natural, we have not found any of its appearance in the literature in this form. However, as it turns out, atoms are precisely the classes of this equivalence relation.}, one can define the syntactic \emph{left} congruence ${}_L\sim$ of a language defined dually as
\[x{}_L\!\sim y\ \Leftrightarrow\ (\forall z:\ zx\in L\Leftrightarrow zy\in L).\]
The \emph{reversal} of a word $w=a_1\ldots a_n$ is the word $w^R=a_n\ldots a_1$, and the reversal of the language $L$ is $L^R=\{w^R:w\in L\}$.
Then obviously, $x\ {}_L\!\sim y$ if and only if $x^R\sim_{L^R} y^R$, since $zx\in L$ holds iff $x^Rz^R\in L^R$.
Hence, classes of the syntactic left congruence are precisely the reversals of the classes of the syntactic right congruence classes of $L^R$.

\section{Atoms of a regular language}
Let $L\subseteq \Sigma^*$ be a regular language and $M=(Q,\Sigma,\delta,q_0,F)$ be its minimal automaton.
Let $n$ stand for $|Q|$, the \emph{state complexity} of $L$.
An \emph{atom} of $L$, as defined in~\cite{DBLP:journals/ijfcs/BrzozowskiT13}, is a nonempty language of the form
\[A_S\quad=\quad\mathop\bigcap\limits_{q\in S}L_q\ \cap\ \mathop\bigcap\limits_{q\notin S}\overline{L_q},\]
for some $S\subseteq Q$. Here $\overline{X}$ stands for complementation with respect to $\Sigma^*$, i.e. $\Sigma^*-X$.
It is clear that $L$ has at most $2^n$ atoms.

That is, a word $w$ is in $A_S$ if $qw\in F$ iff $q\in S$, or equivalently, if $Sw\subseteq F$ and $\overline{S}w\subseteq\overline{F}$.
(For $X\subseteq Q$, $\overline{X}$ denotes $Q-X$.)
An immediate consequence of this characterization is that the atoms of a language are precisely the classes of its syntactic left congruence.
Indeed, first observe that each word $u$ belongs to precisely one atom $A_S$ -- to which $S=\{q:qu\in F\}=\{q_0w:wu\in L\}$.
Hence, $u$ and $v$ belong to the same atom $A_S$ iff $S=\{q_0w:wu\in L\}=\{q_0w:wv\in L\}$ iff $u\ {}_L\!\sim v$.
Thus, atoms are in a one-to-one correspondence with the states of the minimal automaton of $L^R$, in particular
the number of atoms of $L$ coincides with the state complexity of $L^R$.

In~\cite{brzozoMaximally, DBLP:journals/ijfcs/BrzozowskiT13}, the authors achieved results on properties of atoms such as the number and state complexity of individual atoms,
via studying the ``\'atomaton'' of $L$, which is a nondeterministic automaton, actually being isomorphic to the
reversal of the determinized reversal of $M$.
In this paper we suggest another way to study atoms and reprove the characterization of the so-called maximally atomic languages.

To achieve this, we define a modified power set automaton, the \emph{disjoint power square (DPS) automaton} $\mathrm{DPS}(M)=(Q',\Sigma,\Delta,p,F')$ of $M$ as follows: 
\begin{itemize}
\item $Q'\subseteq \bigl(P(Q)\times P(Q)\bigr)\cup\{\bot\}$ consists of the state pairs $(S,T)$ for $S,T\subseteq Q$
  with $S\cap T=\emptyset$, and a sink state $\bot$.
\item $\Delta(\bot,a)=\bot$ and \[\Delta((S,T),a)=\begin{cases}(Sa,Ta)&\textrm{if }Sa\cap Ta=\emptyset\\\bot&\textrm{otherwise}\end{cases}.\]
\item $p$ is an arbitrary state.
\item $F'=P(F)\times P(\overline{F})$, that is, $\{(S,T):S\subseteq F,T\subseteq\overline{F}\}$.
\end{itemize}

It is clear that for any $S,T\subseteq Q$, $S\cap T=\emptyset$ and $w\in \Sigma^*$ we have
\[(S,T)w=\begin{cases}(Sw,Tw)&\hbox{if }{Sw}\cap {Tw}=\emptyset\\\bot&\textrm{otherwise}\end{cases}.\]
Also, since a transformation cannot increase the size of its domain, and the domain of a transformation is empty iff its range is empty,
we get that $(S',T')=(S,T)w$ for some $(S,T)$ and $w$ only if $|S'|\leq |S|$, $|T'|\leq |T|$ with $S'$ ($T'$, resp.) being empty iff
so is $S$ ($T$, resp.)

Employing the notion from~\cite{brzozoMaximally}, we introduce to each $0\leq k\leq n$ the set of  $(n,k)$-type states as follows:
\begin{itemize}
\item $(n,n)$-type states are those of the form $(S,\emptyset)$ with $\emptyset\neq S\subseteq Q$;
\item $(n,0)$-type states are those of the form $(\emptyset,S)$ with $\emptyset\neq S\subseteq Q$;
\item $(n,k)$-type states for $1\leq k<n$ are those of the form $(S,T)$ with $1\leq |S|\leq k$ and $1\leq |T|\leq n-k$, $S\cap T=\emptyset$
  as well as $\bot$.
\end{itemize}

Let $L_{S,T}$ stand for the language accepted by $(Q',\Sigma,\Delta,(S,T),F')$. Then, $w\in L_{S,T}$ iff $Sw\subseteq F$ and $Tw\subseteq \overline{F}$.
In particular, nonempty languages of the form $L_{S,\overline{S}}$ are precisely the (nonempty) atoms $A_S$ of $L$ and thus every atom of $L$
is recognizable in $\mathrm{DPS}(M)$.

Given $S\subseteq Q$, let us call the automaton $M_S=(Q_S, \Sigma, \Delta|_{Q_S}, (S,\overline{S}),F'\cap Q_S)$ the \emph{support automaton} for $S$,
where $Q_S$ consists of the $(n,|S|)$-type states of $\mathrm{DPS}(M)$. Observe that $M_S$ is well-defined since for any $(n,k)$-type state $p$ and
letter $a\in\Sigma$, $\Delta(p,a)$ is also an $(n,k)$-type state.
(Note that in particular $\bot$ is never reachable
from any state of the form $(S,\emptyset)$ or $(\emptyset,S)$ since that would imply $Sw\cap \emptyset w\neq \emptyset$ for some $w$ which 
is nonsense since $\emptyset w=\emptyset$. Hence $\bot$ is not an $(n,n)$ nor an $(n,0)$-type state.)

Hence, $M_S$ is a subautomaton of $\mathrm{DPS}(M)$ recognizing $A_S$, thus
the number of $(n,|S|)$-type states provides an upper bound for the state complexity of $A_S$.

Following~\cite{brzozoMaximally} we define the function $\Psi(n,s)$ as
\[\Psi(n,s)\quad=\quad\begin{cases}2^{n}-1&\textrm{if }n=s\textrm{ or }s=0;\\ 1+\mathop\sum\limits_{k=1}^s\mathop\sum\limits_{\ell=1}^{n-s}\left(n\atop k\right)\left({n-k}\atop\ell\right)&\textrm{otherwise.} \end{cases}\]
This function provides such an upper bound:
\begin{proposition}
The \emph{maximal} number of reachable states from $(S,\overline{S})$ in $\mathrm{DPS}(M)$,
or equivalently, the number of states in $M_S$ is $\Psi(n,|S|)$.

Hence, $\Psi(n,|S|)$ is an upper bound for the state complexity of $A_S$.

\end{proposition}
\begin{proof}
As we already argued, from $(S,\overline{S})$ only $(n,|S|)$-type states are reachable.

When $S\in\{Q,\emptyset\}$, there are $2^n-1$ states that are of $(n,|S|)$-type, handling the cases $n=s$ and $s=0$.

For $\emptyset\neq S\subsetneq Q$, the set of $(n,|S|)$-type states consist of the state $\bot$ (that's where the $1+$ part comes from)
and of state pairs $(S',T')$ with $1\leq |S'|\leq|S|$ and $1\leq |T'|\leq n-|S|$, with $S'$ and $T'$ being disjoint.
We can enumerate the number of those pairs by first choosing the size $1\leq k\leq |S|$ of $S'$, then the size $1\leq \ell\leq n-|S|$
of $T'$, then $k$ elements out of the $n$ states for $S'$, and another $\ell$ elements out of the remaining $n-k$ elements (in order to
ensure disjointness) for $T'$. Summing up we get the formula given in $\Psi(n,|S|)$ as an upper bound.
\end{proof}

In~\cite{brzozoMaximally}, an atom $A_S$ is said to have \emph{maximal complexity}, or is simply called maximal, if its state complexity is $\Psi(n,|S|)$.
That is, if every $(n,|S|)$-state is indeed reachable from $(S,\overline{S})$ and are pairwise distinguishable.
A language (of state complexity at least $2$) is called \emph{maximally atomic} if it has $2^n$ atoms, each being maximal.

We note that if the syntactic semigroup of $L$ contains every transformation $f:Q\to Q$, then $L$ is maximally atomic.
Indeed, let $M=(Q,\Sigma,\delta,q_0,F)$ be an automaton recognizing some language $L$ with
$\emptyset\neq F\subsetneq Q$ (thus $|Q|>1$) and $\mathcal{T}(M)=\mathcal{T}_Q$.
Then for any subset $S$ of $Q$ there is a transformation $f_S$ induced by some word $w_S$
mapping members of $S$ inside $F$ and members of $\overline{S}$ inside $\overline{F}$,
hence $w_S$ is in $A_S$, thus $L$ has all the possible atoms. Moreover, whenever $(S,\overline{S})$ is a state
of $\mathrm{DPS}(M)$ and $(S',T')$ is an $(n,|S|)$-type state, then there is a mapping $f$ mapping $S$ surjectively onto $S'$ and $\overline{S}$ onto $T'$,
thus all the $(n,|S|)$-type states are reachable from each $(S,\overline{S})$ (observe that any constant mapping puts
any $(S,T)$ with nonempty $S$ and $T$ to $\bot$, so $\bot$ is also reachable from $(S,\overline{S})$ when $0<|S|<n$).

We still have to check that whenever $(S,T)$ and $(S',T')$ are different states of $\mathrm{DPS}(M)$, then they are distinguishable:
indeed, if $S\not\subseteq S'$, then $f_{S'}$ brings $(S',T')$ to a final state and $(S,T)$ to a non-final state,
analogously if $S'\not\subseteq S$, then $f_S$ distinguishes the two states. Finally, if $S=S'$, then either
$f_{\overline{T}}$ or $f_{\overline{T'}}$ makes the distinction. Also, $\bot$ is the only empty state, being distinguishable from any $(S,T)$.

Hence we have shown the following result:
\begin{proposition}
Let $L$ be a regular language with $M=(Q,\Sigma,\delta,q_0,F)$ being its minimal automaton, $n=|Q|>1$.
Then if $A_S$ is nonempty, its state complexity is upperbounded by $\Psi(n,|S|)$.

Moreover, this bound is tight: if $\mathcal{T}(M)=\mathcal{T}_Q$, then
$L$ has the maximal possible $2^n$ atoms, with $A_S$ having state complexity $\Psi(n,|S|)$ for each $S\subseteq Q$.
\end{proposition}

Let $\mathcal{T}\subseteq \mathcal{P}_Q$ be a semigroup of permutations of $Q$.
Such a semigroup is called \emph{$k$-set-transitive} for an integer $k\leq|Q|$ if for any two sets $S,T\subseteq Q$ with $|S|=|T|=k$
there exists some $f\in\mathcal{T}$ with $f(S)=T$ and is \emph{set-transitive} if it is $k$-set-transitive for each $1\leq k\leq |Q|$.

\begin{proposition}
\label{prop-weak}
Suppose $L$ is a regular language with minimal automaton $M=(Q,\Sigma,\delta,q_0,F)$, $|Q|\geq 3$.
Then for any atom $A_S$ of $L$, if $A_S$ is maximal, then
$\mathcal{P}(M)$ is $|S|$-set-transitive,
$\mathcal{T}(M)$ contains a transformation of rank $n-1$
and each $A_X$ with $|X|=|S|$ is an atom.
\end{proposition}
\begin{proof}
Suppose $A_S$ is maximal. Then, each $(X,\overline{X})$ with $|X|=|S|$ is an $(n,|S|)$-state, thus has to be
nonempty and reachable from $(S,\overline{S})$. Hence $A_X$ is an atom as well.

Moreover, in that case $(X,\overline{X})=(S,\overline{S})w_X$ for some $w_X$, and thus $w_X$ induces a permutation
with $Sw_X=X$. This implies $Xw_X^{n!-1}=S$ and hence, $Xw_X^{n!-1}w_Y=Y$ for any $X,Y$ with $|X|=|Y|$ and $\mathcal{P}(M)$
is indeed $|S|$-set-transitive.

Also, if $(S,\overline{S})w=(X_1,X_2)$, then $|X_1|+|X_2|$ is the rank of $w$.
If $|S|>1$ and $s\in S$, then $(S-\{s\},\overline{S})$ is an $(n,|S|)$-type state, thus it's reachable from $(S,\overline{S})$
by some word $w$,
hence (by $|S-\{s\}|+|\overline{S}|=n-1$) $w$ induces a transformation of rank $n-1$.
If $|S|\leq 1$, then by $n\geq 3$, $|\overline{S}|>1$ and $(S,\overline{S}-\{s\})$ is reachable from $(S,\overline{S})$ by $w$,
again implying that the rank of $w$ is $n-1$.
\end{proof}

Before proceeding, we recall the following theorem:
\begin{theorem}[\cite{andre}, Thm. 3.2]
\label{thm-andre}
if $\mathcal{S}$ is a subsemigroup of $\mathcal{T}_n$
such that $\mathcal{S}$ contains a transformation of rank $n-1$ and the subgroup of permutations of $\mathcal{S}$ is $2$-set transitive, then
$\mathcal{S}$ contains all the singular transformations in $\mathcal{T}_n$.
\end{theorem}
(Note: the theorem in this form is a rather specialized variant, applied to near permutation semigroups of the form $\mathcal{S}=\langle\mathcal{G},\{u\}\rangle$. However, having only one transformation of rank $n-1$ suits our needs in this paper.)

We are ready to show the following result regarding atoms having maximal state complexity:
\begin{proposition}
Suppose $L$ is a regular language with minimal automaton $M=(Q,A,\delta,q_0,F)$, $|Q|\geq 3$ and that whenever $A_S$ is an atom of $L$
then $A_S$ has maximal (i.e. $\Psi(|Q|,|S|)$) state complexity. Then:
\begin{itemize}
\item $\mathcal{P}(M)$ is set-transitive,
\item $\mathcal{T}(M)$ contains at least one transformation of rank $|Q|-1$ and
\item every $A_T$ is an atom (thus $L$ is maximally atomic).
\end{itemize}
\end{proposition}
\begin{proof}
For better readability, we break the proof into several smaller claims.
Let $n$ stand for $|Q|$.
\begin{enumerate}
\item By Proposition~\ref{prop-weak}, if $A_X$ is an atom for some $X\subseteq Q$ (being maximal by the assumption), $\mathcal{P}(M)$
is $|X|$-set-transitive and $\mathcal{T}(M)$ contains at least one transformation of $n-1$. Thus it suffices to show that $A_T$ is
nonempty for each $T\subseteq Q$ (since then $k$-set-transitivity is implied for each $k$).
\item Also, since if $A_X$ is an atom (maximal by assumption), then by Proposition~\ref{prop-weak} each $A_Y$ with $|Y|=|X|$ is
also an atom. Hence it suffices to show that for each $0\leq k\leq n$ there is an atom $A_X$ with $|X|=k$.
\item $A_F$ is an atom for any automaton, since $\varepsilon\in A_F$. Hence there is a word $w$ inducing a transformation of rank $n-1$,
  say $pw=qw=r$ and $Qw=Q-\{s\}$. Note that $1\leq |F|<n$.
\item Suppose $A_X$ is an atom with $1\leq |X|<n-1$. Let $Y\subseteq Q$ be a set with $|Y|=|X|$, $r\in Y$ and $s\notin Y$
  (such $Y$ exists by the cardinality constraint). Then, $(Y,\overline{Y}-\{s\})$ is an $(n,|X|)$-type state,
  thus it is in particular nonempty by maximality of $A_X$ (since in $M_X$, the only empty state is $\bot$).
  Hence, for $Z=\{p':p'w\in Y\}$ we have $|Z|=|X|+1$ (by $pw=qw=r\in Y$, $s\notin Y$ and $w$ having rank $n-1$).
  By $(Z,\overline{Z})w=(Y,\overline{Y}-\{s\})$, the latter being a nonempty state, $(Z,\overline{Z})$ is nonempty as well, and
  $A_Z$ is an atom of $L$.
  
  Thus if $A_X$ is an atom with $1\leq |X|<n-1$, then so is some $A_Z$ having size $|Z|=|X|+1$.
  This in turn implies $A_Y$ being an atom for each $Y$ with $|F|\leq |Y|\leq n-1$.
  
\item Suppose $A_X$ is an atom with $1<|X|<n$. Let $Y\subseteq Q$ be a set with $|Y|=|X|$, $s\in Y$, $r\notin Y$ (such $Y$ exists).
  Then $(Y-\{s\},\overline{Y})$ is an $(n,|X|)$-type state.
  Hence for $Z=\{p':p'w\in Y\}$ we have $|Z|=|X|-1$ and again, by $(Z,\overline{Z})w=(Y-\{s\},\overline{Y})$
  we get that $(Z,\overline{Z})$ is nonempty, thus $A_Z$ is also an atom.
  
  Hence $A_Y$ is an atom for each $Y$ with $1\leq |Y|\leq |F|$.
\item Hence, every $A_X$ with $X\notin\{\emptyset,Q\}$ is an atom, thus by Proposition~\ref{prop-weak} $\mathcal{P}(M)$ is $k$-set-transitive
for each $1\leq k<n$. Since $n$-transitivity is vacuously satisfied,
$\mathcal{P}(M)$ is set-transitive then.
\item It remains to show that $A_\emptyset$ and $A_Q$ are atoms as well.
Applying Theorem~\ref{thm-andre} we get that $\mathcal{T}(M)$ contains all the transformations having rank less than $n$.
Thus in particular, $Qu=F$ and $Qv=\overline{F}$
for some $u$ and $v$, hence $(Q,\emptyset)$ and $(\emptyset,Q)$ are nonempty, thus $A_Q$ and $A_\emptyset$ are also atoms.
\end{enumerate}
\end{proof}

We can also show the converse direction:
\begin{proposition}
Suppose $M=(Q,\Sigma,\delta,q_0,F)$ is the minimal automaton of $L$ with $n=|Q|>2$ such that
$\mathcal{P}(M)$ is set-transitive and
$\mathcal{T}(M)$ contains at least one transformation of rank $n-1$.

Then $L$ is maximally atomic.
\end{proposition}
\begin{proof}
If $\mathcal{P}(M)$ is set-transitive, then it is also $2$-set-transitive. Since by assumption $\mathcal{T}(M)$ contains a transformation
of rank $n-1$, we get by Theorem~\ref{thm-andre} that $\mathcal{T}(M)$ contains all the singular transformations.
Hence for any set $S\subseteq Q$ we get that from $(S,\overline{S})$
\begin{itemize}
\item each state $(X,\overline{X})$ with $|S|=|X|$ is reachable by $|S|$-set transitivity;
\item each $(n,|S|)$-type state $(S',T')$ with $|S'|+|T'|<n$ is reachable since in that case there exists a singular transformation
$f$ mapping $S$ onto $S'$ and $\overline{S}$ onto $T'$;
\item if $1\leq |S|<n$, then also $\bot$ is reachable by any constant (thus singular) transformation.
\end{itemize}
Hence every $(n,|S|)$-type state is reachable from $(S,\overline{S})$ for each $S$.
We only have to show that each $(S,T)$ is nonempty and distinguishable from any other $(S',T')$.
Indeed, by $n>2$ we have $1\leq |F|<n$, thus there exist states $p\in F$, $q\notin F$. Hence the function induced by some word
$u_S$ which maps each $s\in S$ to $p$ and each $s\notin S$ to $q$, maps $(S,T)$ to the final state $(\{p\},\{q\})$.
Also, let $(S',T')\neq(S,T)$. If $S\neq S'$, then either $u_S$ (if $S'\not\subseteq S$) or $u_{S'}$ (if $S\not\subseteq S'$)
distinguishes the two states, while if $T\neq T'$, then either $u_{\overline{T}}$ or $u_{\overline{T'}}$ does.
(The argument works if any of the sets $S$, $T$, $S'$ and $T'$ is empty as well).
\end{proof}

As a corollary we proved the following characterization of maximally atomic regular languages:
\begin{theorem}
Suppose $L$ is a regular language with minimal automaton $M=(Q,\Sigma,\delta,q_0,F)$, $|Q|\geq 3$.
Then the following are equivalent:
\begin{enumerate}
\item $L$ is maximally atomic;
\item every atom of $L$ has maximal complexity;
\item $\mathcal{T}(M)$ contains a transformation of rank $|Q|-1$, and $\mathcal{P}(M)$ is set-transitive.
\end{enumerate}
\end{theorem}

\section{Conclusions}

We introduced a particular automaton, the so-called ``disjoint power set automaton'' of a regular language
and using this notion we gave a combinatorial proof of the characterization of maximally atomic languages,
originally showed in~\cite{brzozoMaximally}.
It is an interesting question whether the DPS automaton of a language can have other uses as well.

\bibliographystyle{plainnat}

\bibliography{biblio}{}

\end{document}